\documentclass[nofootinbib,twocolumn,letterpaper]{revtex4-2}

\usepackage{hyperref}
\usepackage{latexsym,amsmath,amssymb,amsfonts,graphicx,color,amsthm}
\usepackage{enumerate}
\usepackage{braket}
\usepackage{adjustbox}
\usepackage{footnote}
\usepackage[dvipsnames]{xcolor}

\newcommand{\cA}{\mathcal{A}}
\newcommand{\cB}{\mathcal{B}}

\newcommand{\cF}{\mathcal{F}}
\newcommand{\cG}{\mathcal{G}}
\newcommand{\cH}{\mathcal{H}}

\newcommand{\cM}{\mathcal{M}}

\newcommand{\cQ}{\mathcal{Q}}

\newcommand{\cS}{\mathcal{S}}

\newcommand{\cX}{\mathcal{X}}

\newcommand{\Id}{\mathbb{I}}
\newcommand{\tr}{\text{tr}}

\newtheorem{theorem}{Theorem}
\newtheorem{proposition}{Proposition}
\newtheorem{lemma}[theorem]{Lemma}
\newtheorem{corollary}[theorem]{Corollary}
\newtheorem{definition}{Definition}

\begin{document}

\title{On optimal cloning and incompatibility}
\author{Arindam Mitra$^{1,2}$}
\affiliation{$^1$Optics and Quantum Information Group, The Institute of Mathematical Sciences,
C. I. T. Campus, Taramani, Chennai 600113, India.\\
$^2$Homi Bhabha National Institute, Training School Complex, Anushakti Nagar, Mumbai 400094, India.}

\author{Prabha Mandayam}
\affiliation{Department of Physics, Indian Institute of Technology Madras, Chennai - 600036, India.}
\date{\today}

\begin{abstract}
We investigate the role of symmetric quantum cloning machines (QCMs) in quantifying the mutual incompatibility of quantum observables. Specifically, we identify a cloning-based incompatibility measure whereby the incompatibility of a set of observables maybe quantified in terms of how well a uniform ensemble of their eigenstates can be cloned via a symmetric QCM. We show that this new incompatibility measure $\cQ_{c}$ is {\it faithful} since it vanishes only for commuting observables. We prove an upper bound for $\cQ_{c}$ for any set of observables in a finite-dimensional system and show that the upper bound is attained if and only if the observables are mutually unbiased. Finally, we use our formalism to obtain the optimal quantum cloner for a pair of qubit observables. Our work marks an important step in formalising the connection between two fundamental concepts in quantum information theory, namely, the no-cloning principle and the existence of incompatible observables in quantum theory.

\end{abstract}

\maketitle

\section{Introduction}\label{sec:intro}

Perfect cloning of an arbitrary state is impossible in quantum mechanics~\cite{wootters}. This is the essence of the celebrated \emph{no-cloning theorem} in quantum information theory, which marks a fundamental departure from the nature of information in the classical world. More precisely, the theorem states that it is impossible to copy a pair of non-orthogonal pure states via any quantum operation. Cloning is in fact a special case of a more general information processing task, namely, that of \emph{broadcasting}. Broadcasting extends the idea of cloning to mixed states by removing the restriction that the different clones need to be independent; it only requires that the marginals of the separate systems be the same as the original state. The quantum \emph{no-broadcasting} theorem thus extends the pure-state no-cloning result to state that non-commuting quantum density operators cannot be broadcast perfectly~\cite{no-broadcast96}.

The no-cloning and no-broadcasting theorems are indeed related to other fundational physical principles including the no-singaling principle~\cite{dieks82} and the uncertainty principle~\cite{scarani05}. They also have other important operational consequences. For example, perfect cloning would allow for perfect quantum state estimation using a single copy of an unknown state. Indeed, in the limit of infinitely many clones, the equivalence between quantum cloning and state estimation is well established~\cite{acin}. Similar results are known for the asymptotic  broadcasting scenario as well~\cite{chiribella, brandao}. 



Although perfect cloning or broadcasting is impossible, approximate or imperfect cloning of arbitrary quantum states is still possible, as first demonstrated in~\cite{buzek_hillery96}. Several approaches towards designing cloning devices have been since studied in the literature, the most well-known of which is perhaps the symmetric universal cloner~\cite{buzek_hillery}, which produces two identical copies of an arbitrary quantum state in any dimension, with \emph{optimal} fidelity. This construction of a $1 \rightarrow 2$ universal, symmetric quantum cloning machine (QCM) has further been extended to obtain optimal, universal, symmetric QCMs which  transform $n$ copies of a given quantum state into $m>n$ copies~\cite{gisin_massar, werner}.  Asymmetric cloning, in which the clones do not all have the same fidelity, has also been studied and trade-offs between the fidelities of the different clones have been discussed in~\cite{niu_griffiths,cerf}. Approximate cloning machines have interesting applications, for example, in estimating weak values of non-commuting observables~\cite{hofmann}. We refer to~\cite{cloning_review} for a recent review on QCMs and their applications. 


Starting with the uncertainty principle~\cite{heisenberg}, several measures have been proposed to quantify the extent of incompatibility of a set of quantum observables. For example, entropic uncertainty bounds are often used as measures of incompatibility~\cite{wehner_winter}. Other measures based on operational constraints have also been proposed for arbitrary sets of quantum observables~\cite{incompatibility_BM, PM_MS1}. 

The fact that only ensembles of commuting states can be cloned (or broadcast) perfectly, suggests a natural connection between the no-cloning/no-broadcasting theorems and incompatibility in quantum theory. For example, the universal QCM has been used to show that only infinite dimensional spaces admit maximal incompatibility, where incompatibility is quantified via joint measurability~\cite{heino}. Furthermore, it has been argued that the optimal cloning fidelity for a set of non-commuting density operators on a $d$-dimensional space gives a lower bound on the degree of incompatibility possible for a $d$-dimensional quantum state~\cite{heinosaari}. More recently, the results of~\cite{heino} have been generalised by making use of approximate cloning machines to characterize the idea of \emph{compatibility dimension}~\cite{nechita_2020}. 

We may note here that the cloning-based incompatibility characterizations described above, typically model the cloning device as a quantum channel. Observables are measured on local, imperfect copies of a state and the statistics of these imperfect clones are compared with the statistics of the marginals of a joint measurement given by the dual of the cloning channel~\cite{heinosaari2016}. Our operational set-up is rather different and motivated by quantum cryptographic protocols, where ensembles of eigenstates of (incompatible)  observables are used for secure communication~\cite{gisin97}. We use the quantum cloning device to clone the ensemble of eigenstates corresponding to a set of Hermitian operators and hence make a quantitative connection between imperfect cloning of the ensemble and the incompatibility of the corresponding physical observables. 

Specifically, we examine the role of a symmetric $1\rightarrow 2$ QCM in quantifying the mutual incompatibility of quantum observables. Corresponding to a given set of physical observables -- associated with a set of Hermitian operators -- we consider the ensemble corresponding to a uniform mixture of their eigenstates. We propose that the incompatibility of a set of observables maybe quantified in terms of the optimal cloning fidelity of the associated ensemble of eigenstates. The more incompatible the set is, the smaller is the cloning fidelity. 

We show that our cloning-based measure is a {\it faithful} measure of incompatibility, in the sense that it vanishes iff the observables commute. We also obtain a tight upper bound on this new measure and show that the maximum value is attained only for a set of mutually unbiased bases. Our work provides the first quantitative statement relating the operational task of quantum cloning of a given ensemble of eigenstates and the mutual incompatibility of the corresponding set of Hermitian operators. The bounds proved here might therefore prove useful in the context of cryptographic tasks such as quantum key distribution, where the eavesdropper could carry out a quantum cloning attack on the signal states.

The rest of this paper is organized as follows. In~\ref{sec:prelim} we briefly review the $1\rightarrow 2$ symmetric QCM, and describe the optimal cloning fidelity of an ensemble of states. In Sec.~\ref{sec:cloning_measure}, we define the cloning-based incompatibility measure $\cQ_{c}$ and prove tight upper and lower bounds. Specifically, we calculate the measure for a set of mutually unbiased bases (MUBs) in Sec.~\ref{sec:Qc_mub} and use this to obtain an upper bound on $\cQ_{c}$ in Sec.~\ref{sec:Qc_ub}. We use our formalism to characterize the optimal QCM for a pair of qubit observables in Sec.~\ref{sec:qubit}. In Sec.~\ref{sec:general}, we discuss the possibility of extending our measure beyond a set of rank-one projectors, to quantify the incompatibility of a general set of quantum measurements. Finally, we summarize our contributions and discuss the future outlook in Sec.~\ref{sec:concl}.


\section{Preliminaries}\label{sec:prelim}

In this section, we briefly recall the mathematical description of the symmetric, universal quantum cloner and analyze the optimal cloning fidelity of an ensemble of states.



\subsection{ Symmetric Quantum Cloning Machine}\label{sec:sym_cloning}

We briefly review the well known symmetric quantum cloning machine (QCM). A cloning operation is said to be \emph{symmetric} if the fidelity between input and output states are equal for \emph{all} of the clones; in other words all the copies have identical fidelity with the original state. The cloner is further said to be \emph{universal} if it achieves a constant, state-independent fidelity between the original (input) and the cloned (output) states. In what follows, we will consider the symmetric cloner without loss of generality, because it is always possible to obtain an optimal symmetric cloner by taking linear combinations of optimal asymmetric cloners~\cite{duan17}.


The symmetric QCM is a unitary transformation on a tripartite system $\cH_{A}\otimes\cH_{B}\otimes\cH_{C}$, where $\cH_{A}$ denotes the `input' system containing the state to be cloned, $\cH_{B}$ corresponds to the ancilla system into which the original state is cloned, and $\cH_{C}$ denotes the Hilbert space associated with the cloner. Suppose $\cH_{A}$ and $\cH_{B}$ are Hilbert spaces of dimension $d$. Let $\cB \equiv \{\ket{i}_A, 1\leq i \leq d\}$ denote a fixed orthonormal basis in $\cH_{A}$. The symmetric QCM is then defined via the linear transformation~\cite{buzek_hillery},
\begin{eqnarray}
&& \ket{i}_A\ket{0}_B\ket{X}_C \longrightarrow   \; p\ket{i}_A\ket{i}_B\ket{X_i}_C  \nonumber \\
&& + \; q \, \sum_{j\neq i} \left(\ket{i}_A\ket{j}_B + \ket{j}_A\ket{i}_B \right)\ket{X_j}_C , \label{eq:sym_clone}
\end{eqnarray}
where, $\ket{X}_C$ is a fixed (normalized) state of system $\cH_{C}$, $\ket{0}_B$ is a blank ancilla state which is transformed to a clone of $|i\rangle_{A}$, and $\{ \ket{X_i}_C\}$ are the fixed orthonormal basis vectors of cloning machine. Clearly, the joint tripartite state is symmetric in the first two systems. The real coefficients $p$ and $q$ must satisfy the following relation in order to ensure unitarity of the QCM transformation:
\begin{equation}
p^2+2(d-1)q^2=1 . \label{eq:unitarity}
\end{equation}
Equivalently, the QCM can be modeled as the quantum channel $\Phi$ mapping density operators on system $\cH_A$ to density operators on $\cH_A \otimes \cH_B$, defined as
$\Phi(\rho) = {\rm tr}_{C} \left[ U(\rho \otimes |0\rangle\langle 0| \otimes |X\rangle \langle X|)U^{*}\right] $ for all states $\rho$ on system $\cH_A$, where $U$ is unitary operator associated with the cloning machine. In the rest of the paper, we will use the linear transformation in Eq.~\eqref{eq:sym_clone} to describe the action of the symmetric QCM.  

Suppose $\ket{\phi}=\sum_i\alpha_i\ket{i}$ is the state in $\cH_{A}$ to be cloned. After applying the QCM transformation, the clones are given by the output density matrices,
\begin{align}
& (\rho_{\rm out})_A = (\rho_{\rm out})_B \nonumber \\
 &= \sum^d_{i=1}\mid\alpha_i\mid^2 \left( p^2+(d-2)q^2 \right) \ket{i}\bra{i}\nonumber \\ 
 &+\sum^d_{{i,j=1},{i\neq j}}\alpha_i\alpha^*_j \left( 2pq+(d-2)q^2 \right)\ket{i}\bra{j} + q^2 I . \label{eq:rho_out}
\end{align}
The QCM defined above is thus symmetric since the two clones are identical. 

Note that the class of symmetric QCMs defined by Eq.~\eqref{eq:sym_clone} is characterized by the choice of basis $\cB$, as well as the parameters $(p,q)$, subject to the unitarity constraint in Eq.~\eqref{eq:unitarity}. In other words, the final cloned state $\rho_{\rm out}$ corresponding to a given input state $|\phi\rangle$ varies depending on the choice of basis $\cB$ with respect to which the QCM is defined, apart from its explicit dependence on the choice of the parameters $(p,q)$.  We will refer to the basis $\cB$ with respect to which the cloner is defined, as the \emph{cloning basis}.

We may note two special cases here. The first is the limiting case $q=0$. Unitarity implies that $p=1$, implying that the QCM operation in Eq.~\eqref{eq:sym_clone} becomes,
\[ \ket{i}_A\ket{0}_B\ket{X}_C \longrightarrow  \ket{i}_A \ket{i}_B\ket{X_i}_{C}. \] 
In other words, the resulting QCM is a perfect cloner for the specific choice of basis $\cB\equiv \{\ket{i}_A\}$.  Now, consider a density matrix $\rho$, expressed in the basis $\cB$ as, $\rho = \sum_{ij}a_{ij}\ket{i}\bra{j}$. The cloned density matrices obtained using the quantum cloner with $q=0$ are of the form, 
\[ (\rho_{\rm out})_{A} = (\rho_{\rm out})_{B} =\sum_{i} a_{ii}\ket{i}\bra{i}. \]
Thus the clones are identical to the post-measurement state obtained after a projective measurement of the input state $\rho$ in the basis $\cB$, thus showing that the action of the QCM with $q=0$ is identical to that of a projective measurement followed by state reconstruction. 

The second limiting case to note is that of the \emph{universal}, symmetric cloner, which results when we impose the additional constraint $p^{2} = 2pq$~\cite{buzek_hillery}. This constraint ensures that the cloned states are of the form,
\begin{equation}
(\rho_{\rm out})_{A} = (\rho_{\rm out})_{B} = s |\phi\rangle\langle \phi| + \frac{(1-s)}{d} I, \label{eq:depolarizing}
\end{equation}
where $s = \frac{d+2}{2(d+1)}$ is a purely dimension-dependent factor, independent of the input state $\ket{\phi}$. This ensures universality in the sense that \emph{all} input states are cloned with the same fidelity. The universal, symmetric QCM described here is known to be \emph{optimal}, in the sense that it achieves the maximum possible value of $s$ for any dimension~\cite{werner}.

\subsection{Optimal cloning fidelity of an ensemble}\label{sec:opt_cloning}

Using the symmetric QCM described by Eqs.~\eqref{eq:sym_clone} and~\eqref{eq:unitarity}, we would like to quantify the average cloning fidelity obtained for an ensemble of states. Specifically, suppose the ensemble of states 
\begin{equation}
\cS \equiv \{|\psi_{m}\rangle\langle \psi_{m}|, 1\leq m \leq M\},  \label{eq:ensemble}
\end{equation}
is cloned using the QCM described in Eq.~\eqref{eq:sym_clone}. Let $\rho_{m}=\ket{\psi_{m}}\bra{\psi_{m}}$ and $(\rho_{m})_{\rm out}$ denote the clone corresponding to the state $\rho_{m}$ in the ensemble $\cS$. As before, let $\cB=\{\ket{i}, i = 1,2,\ldots,d\}$ denote the fixed orthonormal basis in which the QCM is defined. The average cloning fidelity $\cF_{\rm avg}(\cS, \cB, p, q)$ for the ensemble $\cS$ is then defined as,

\begin{equation}
\cF_{\rm avg} (\cS, \cB, p,q) = \frac{1}{M}\left(\sum_{m}\tr[ \, (\rho_{m})_{\rm out} \rho_{m} \, ] \right). \label{eq:avg_fidelity}
\end{equation}

Note that the average cloning fidelity $\cF_{\rm avg}$ is indeed a function of the input ensemble for the case of the symmetric QCM under consideration here. The universal cloner, on the other hand, would yield a constant fidelity for all input states.  Furthermore, the average cloning fidelity also depends on the basis $\cB$ with respect to which the QCM is defined, as well as the parameters $p,q$. We further make this dependence explicit by evaluating the average cloning fidelity for a general ensemble $\cS$.

Expanding the state $\ket{\psi_{m}}$ in the basis $\ket{i}$, we have,
\[ \ket{\psi_{m}} = \sum^d_{i} (\alpha_{m})_i\ket{i} . \]
The input states can therefore be written as,
$$\rho_{m}=\ket{\psi_{m}}\bra{\psi_{m}}=\sum_{i=1}^{d}\sum_{j=1}^{M}(\alpha_{m})_i(\alpha_{m})^{*}_{j}\ket{i}\bra{j}. $$

Applying the symmetric QCM on the above state, we get,
\begin{align}
(\rho_{m})_{\rm out}&= \sum_{i=1}^{d}\mid(\alpha_{m})_i\mid^2(p^2+(d-2)q^2)\ket{i}\bra{i}\nonumber\\
& + \sum^d_{i \neq j}(\alpha_{m})_i (\alpha_{m})^{*}_j (2pq+(d-2)q^2)\ket{i}\bra{j}+ q^2 I . \label{eq20}
\end{align}
The fidelity between the input state $|\psi_{m}\rangle$ and its clone can therefore be evaluated as,
\begin{eqnarray}
&& \tr\left[ \rho_{m}(\rho_{m})_{\rm out} \right] \nonumber\\
&=& \sum_{i=1}^{d}\mid (\alpha_{m})_i\mid^4 (p^2+(d-2)q^2)\nonumber\\
&&  + \; \sum^{d}_{i\neq j}\mid (\alpha_{m})_i\mid^2\mid (\alpha_{m})_j\mid^2 (2 pq+(d-2)q^2) + q^2\nonumber\\
&=& A_{m}( p^2+(d-2)q^2 )+ B_{m}( 2pq+(d-2)q^2 )+ q^2 ,\nonumber
\end{eqnarray}
where the state-dependent factors $\{A_{m}\}$ and $\{B_{m}\}$ are defined as,
\begin{eqnarray}
A_{m} &=& \sum_{i=1}^{d}\mid (\alpha_{m})_i\mid^4  , \nonumber \\
B_{m} &=& \sum^d_{i\neq j}\mid (\alpha_{m})_i\mid^2\mid (\alpha_{m})_j \mid^2 . \label{eq:A_B}
\end{eqnarray}
Note that $\{A_{m}\}$ are in effect the participation ratios for the input states $|\psi_{m}\rangle $ in the basis $\cB$. Since $\sum_{i=1}^{d} \vert (\alpha_{m})_i \vert^2 = 1$ for all $1\leq m \leq M$, it follows that,
\[  1 \geq A_{m} = \sum_{i=1}^{d}\mid (\alpha_{m})_i\mid^4 \geq \frac{1}{d}, \; \forall \; m. \] 
Furthermore, since the states $|\psi_{m}\rangle$ are normalised, $A_{m}   + B_{m} = 1$ for all $1\leq m \leq M$, we have,
\begin{equation}
B_{m} =1 - A_{m} \leq 1 - \frac{1}{d}, \; \forall \, m.  \label{eq:B_lm}
\end{equation}
Thus the average cloning fidelity defined in Eq.~\eqref{eq:avg_fidelity} can be explicitly evaluated as,
\begin{eqnarray}
&& \cF_{\rm avg}(\cS, \cB, p, q) \nonumber \\
&=& \frac{1}{M}\bigg[\sum_{m}(A_{m}[p^2+(d-2)q^2]\nonumber\\
&& \hspace{3cm}+B_{m}[2pq+(d-2)q^2]+q^2)\bigg]\nonumber\\
&=& \frac{1}{M}\left[A p^2+ B (2pq) + Mq^2 + (A + B)(d-2)q^2 \right], \nonumber
\end{eqnarray}
where,
\begin{eqnarray}
A(\cS, \cB) &=& \sum_{m=1}^{M}A_{m} = \sum_{m=1}^{M}\sum_{i=1}^{d}\mid (\alpha^{l}_{m})_i\mid^4 \label{eq:A_B2}  \\
B (\cS,\cB) &=& \sum_{m=1}^{M}B_{m} = \sum_{m=1}^{M}\sum^d_{i\neq j}\mid (\alpha_{m})_i\mid^2\mid (\alpha_{m})_j \mid^2 . \nonumber
\end{eqnarray}
The quantities $A$ and $B$ are related as,
\[ B =\sum_{m=1}^{M}(1-A_{m})= M-A,\]
so that the average cloning fidelity can be expressed as,
\begin{eqnarray}
\cF_{\rm avg}(\cS, \cB, p,q) &=& \frac{A (\cS,\cB)}{M} (p^{2}-2pq) \nonumber \\
&+&  2pq +  (d-1)q^{2} .  \label{eq:avg_fid2}
\end{eqnarray}

Once again, we note the two limiting cases here. For $q=0$, $\cF_{\rm avg}(\cS,\cB, p=1,q=0) =\frac{A(\cS,\cB)}{M}$. Since the QCM in this case simply implements a projective measurement in the cloning basis $\cB$, $\frac{A(\cS,\cB)}{M}$ represents the average fidelity gain due to the projective measurement and state reconstruction, studied in~\cite{incompatibility_BM}. For completeness, we have defined the average fidelity obtained via projective measurement and state reconstruction given in Eq.~\eqref{eq:fid_avg1} of Appendix~\ref{sec:Q_avgFid}. We formally state this equivalence below.
\begin{proposition}\label{prop:zero_q}
For an ensemble $\cS$ of $M$ quantum states, the average cloning fidelity obtained using a symmetric QCM with cloning basis $\cB$ and $p=1$, $q=0$, is the same as the average fidelity gained via a measurement and reconstruction procedure, where the measurement is a projective measurement in the basis $\cB$, that is,
\begin{eqnarray}
&& \cF_{\rm avg}(\cS,\cB, p=1,q=0) = \frac{A(\cS,\cB)}{M} \nonumber \\
&& \equiv \mathbb{F}_{\rm avg}(\cS, \cM_{\cB}, \cA_{\cB}), \label{eq:cloning_fid_equality}
\end{eqnarray}
where $\cM_{\cB}$ denotes a measurement in basis $\cB$ and $\cA_{\cB}$ denotes the corresponding reconstruction strategy. 
\end{proposition}

On the other hand, for $p^2=2pq$, 
\begin{equation}
 \cF_{\rm avg}(\cS) = \frac{d+3}{2(d+1)} \equiv \cF_{\rm univ}(d) ,  \label{eq:fopt_univ}
 \end{equation}
independent of $A$ and therefore independent of the ensemble $\cS$ as well as the choice of cloning basis $\cB$. This is the case of the universal symmetric cloning machine, and the corresponding cloning fidelity -- which we denote as $\cF_{\rm univ}(d)$ -- is indeed a constant across input states for a given dimension~\cite{buzek_hillery}.

Finally, we define the \emph{optimal cloning fidelity} for the ensemble $\cS$ as the maximum possible value of the average cloning fidelity when optimized over \emph{all} symmetric QCMs corresponding to different cloning bases $\cB$:
\begin{equation}
\cF_{\rm opt}(\cS) = \max_{p,q}\max_{\cB} \cF_{\rm avg}(\cS,\cB,p,q). \label{eq:opt_cloning_fid1}
\end{equation}
In what follows, we will refer to the basis $\cB_{\rm opt}$ which maximizes the average cloning fidelity as the {\emph optimal cloning basis} and use $(p_{\rm opt},q_{\rm opt})$ to denote the optimal values of the parameters $(p,q)$.

At first glance, computing the optimal cloning fidelity for a general ensemble $\cS$ appears to be a daunting task, since it involves a double optimization. However, looking at the form of the average fidelity function in Eq.~\eqref{eq:avg_fid2}, it is possible to do the optimization over $\cB$ first and then perform the optimization over the parameters $p,q$, provided we consider two distinct regimes. When $p^{2} > 2pq$, the optimal cloning basis $\cB_{\rm opt}$ can be fixed by  calculating the maximum value of the quantity $A(\cS,\cB)$, defined as,
\begin{equation}
\max_{\cB}A(\cS,\cB) = A_{\rm max}(\cS) \equiv A(\cS, \cB_{\rm opt}).
\end{equation}  
When $p^{2} < 2pq$, the optimal cloning basis is the one which minimizes the quantity $A$, that is,
\[ \min_{\cB}A(\cS,\cB) = A_{\rm min}(\cS) \equiv A(\cS, \cB_{\rm opt}).\]

Guided by the properties of the general, symmetric QCM discussed in Sec.~\ref{sec:sym_cloning}, we choose to work in the regime where $p^{2} > 2pq$. Specifically, we see from Eq.~\eqref{eq:rho_out} that the larger the value of $q$, larger is the contribution of the maximally mixed state to the output states $\rho_{\rm out}$ of the QCM. Hence, in the regime where $p^{2} < 2pq$, we expect the clones corresponding to different ensembles to be more identical to each other, and the corresponding QCMs may not be good at differentiating between the properties of different input ensembles. 

On the other hand, in the regime where $p^{2} > 2pq$, the contribution of the maximally mixed state is smaller, and we expect the clones corresponding to different input ensembles to reflect the properties of the corresponding ensembles. Thus we define the optimal cloning fidelity as,
\begin{equation}
\cF_{\rm opt}(\cS) = \underset{(p,q): p^{2}>2pq}{\rm max}\max_{\cB} \cF_{\rm avg}(\cS,\cB,p,q), \label{eq:opt_cloning_fid}
\end{equation}
and this forms the basis of our cloning-based incompatibility measure defined in the following section. 

\section{Quantifying Incompatibility via symmetric QCMs}\label{sec:cloning_measure}

Symmetric QCMs provide an operational approach to quantify the mutual incompatibility of a set of quantum observables, based on how well the corresponding eigenstates can be cloned. Specifically, given a set of $N$ observables $\cX = \{ X^1,\ldots,X^N \}$ on a $d$-dimensional Hilbert space, we define a measure of incompatibility in terms of the optimal cloning fidelity obtained using a symmetric QCM for a uniform mixture of their eigenstates. Note that, in what follows, we assume the canonical association of a set of physical observables with a set of Hermitian operators on the space. Furthermore, we assume a set of nondegenarate Hermitian operators so that the ensemble $\cS$ essentially comprises rank-one projection operators. In Sec.~\ref{sec:general}, we indicate as to how the measure defined here can be generalized beyond the case of nondegenrate Hermitian operators. 

\subsection{Operational Motivation and Definitions}\label{sec:motivation}

The motivation for broadcasting an ensemble of eigenstates of observables comes
from quantum cryptography. In a typical quantum key distribution (QKD) protocol, one party (Alice) encodes classical information in eigenstates of incompatible observables and the receiver (Bob) decodes by making suitable measurements at his end. Such a protocol is considered secure against eavesdropping, due to the fact  that eigenstates corresponding to incompatible observables cannot be cloned or discriminated perfectly~\cite{gisin97}. However, the eavesdropper could always implement a \emph{cloning attack}, making use of a QCM to make imperfect clones of the states being sent by Alice. It is therefore important from a cryptographic point of view, to quantify how best a quantum cloning device can clone such ensembles of eigenstates. This then becomes a natural operational setting to quantify the incompatibility of the corresponding observables as well.

Formally, we define the eigenstate ensemble $\cS$ corresponding to a set of observables $\cX$, as follows. 
\begin{definition}[Eigenstate ensemble]\label{def:ensemble}
Let $\cX = \{ X^1,\ldots,X^N \}$ be a set of $N$ observables  on a $d$-dimensional Hilbert space, and let $\{\ket{\psi^{l}_{m}}\}$ denote the eigenbasis of observable $X^{l}$. Then, the eigenstate-ensemble $\cS$ corresponding to the set $\cX$ is defined as the set of eigenstates of the observables in the set, with each eigenstate picked with equal probability. That is,
\begin{equation}
 \cS = \left\{ \frac{1}{Nd}, \ket{\psi^{l}_{m}}\right\}, \label{eq:eig_ensemble}
\end{equation}
 with  $1 \leq m \leq d$ and $1\leq l\leq N$.  
\end{definition}

Note that such an ensemble can be cloned perfectly if and only if it is comprised of mutually orthogonal states. The corresponding set of observables would then have to form a commuting set of observables. However, for a general set of observables $\cX$, the corresponding eigenstate ensemble $\cS$ will have some non-orthogonal states. Such an ensemble cannot be cloned with an optimal cloning fidelity of $1$. We may therefore use the deviation of the optimal cloning fidelity of their eigenstate ensemble from unity as a means to quantify the mutual incompatibility of a set of quantum observables. Our intuition suggests that the smaller the optimal cloning fidelity of their eigenstate ensemble, the larger is the mutual incompatibility of a set of observables. This intuition leads to a cloning-based incompatibility measure defined below.

Recall that the average cloning fidelity $\cF_{\rm avg}(\cS,\cB, p, q)$ for an ensemble $\cS$ using a symmetric QCM with parameters $(p,q)$ defined with respect to basis $\cB$, is given by the expression in Eq.~\eqref{eq:avg_fidelity}. We may now define the cloning-based incompatibility measure $\cQ_{c}$ as follows.
\begin{definition}[Incompatibility Measure]
For any set $\cX$ of $N$ observables on a $d$-dimensional Hilbert space, the incompatibility $\cQ_{c}$ is defined as,
\begin{equation}
\cQ_{c} (\cX) = 1 - \cF_{\rm opt}(\cS),\label{eq:Qc_defn}
\end{equation}
where, $\cS$ is the ensemble of eigenstates of the observables in $\cX$ and $\cF_{\rm opt}(\cS)$ is the optimal cloning fidelity defined in Eq.~\eqref{eq:opt_cloning_fid}.
\end{definition}

Before we proceed to analyze this measure in detail, we would like to make the following remarks.
\begin{itemize}
\item[(a)] As mentioned in Sec.~\ref{sec:intro}, there exist nice approaches in the literature to characterize the incompatibility of a pair of POVMs, using approximate  cloning or broadcasting channels~\cite{heino, heinosaari2016}. However, these approaches also suffer from certain drawbacks. For example, while the symmetric \emph{universal} cloner gives rise to a depolarizing channel between the input states and the cloned state (see Eq.~\eqref{eq:depolarizing}), a general symmetric cloner does not always have depolarizing form; rather it may correspond to an arbitrary noise channel between the input and cloned states. Correspondingly, in the Heisenberg picture, a general symmetric cloning channel would transform observables in such a way that the output observables are not simply unsharp versions of the original observables. Hence, optimizing over all symmetric cloners might not lead to a sensible measure in this scenario. Furthermore, it is still not known if the incompatibility bounds obtained using a universal, symmetric cloning channel are tight~\cite{heinosaari}. 

In a departure from this approach, we aim to quantify how well a symmetric cloning device can clone the ensemble of eigenstates associated with a set of physical observables. In this scenario, we are able to obtain tight upper and lower bounds on such a cloning-based incompatibility. Our measure therefore provides an alternate perspective, on how quantum cloning devices can be used to quantify incompatibility.

\item[(b)] It is worth noting here that while the cloning-based incompatibility measure is similar in spirit to the measure $\cQ$ proposed in~\cite{incompatibility_BM}, the two measures are operationally quite different. As discussed in Appendix~\ref{sec:Q_avgFid}, the measure $\cQ$ captures the incompatibility of a set of observables based on how well their eigenstate ensemble maybe reproduced via a measurement-and-reconstruction protocol, where the measurement strategy could include POVMs as well. On the other hand, our cloning-based measure $\cQ_{c}$ captures the incompatibility of a set of observables as reflected by how well their eigenstate ensemble can be cloned, using a symmetric $1\rightarrow 2$ QCM. The optimization is now performed over all symmetric quantum cloners. The fact that the measures $\cQ_{c}$ and $\cQ$ capture different operational notions of incompatibility is reflected in subsequent sections when we compare and contrast the numerical values of these two measures for mutually unbiased observables.

\item[(c)] Finally, we focus on symmetric rather than asymmetric QCMs for the following reason. In asymmetric cloning, the clones do not all have the same fidelity. Such asymmetric cloning machines are indeed useful to derive trade-off relations between the fidelities with which different bases can be clones. However, our goal is to construct a cloning-based measure that captures the non-orthogonality of an ensemble of states corresponding to different bases (or different observables). An asymmetric cloner that acts differently on states belonging to different bases would not suit our purpose. Rather, we need to use symmetric cloners which are not biased to any particular basis (or observable) in a given set. In other words, while symmetric cloning is useful to quantify incompatibility of a set of bases in absolute terms, asymmetric cloning is more useful as a relative measure, to quantify trade-offs between individual bases or observables.
\end{itemize}

\subsection{Properties of the measure $\cQ_{c}$}\label{sec:props}

The following proposition shows that the measure $\cQ_{c}$ is a true measure of incompatibility, since it attains a trivial value only for a set of commuting observables.

\begin{lemma}\label{lem:Qc_lb}
The cloning-based incompatibility measure satisfies $\cQ_{c}(\cX) \geq 0$, with equality attained iff all the observables in $\cX$ commute.
\end{lemma}
\begin{proof}
The proof simply follows from the fact that a set of states can be cloned perfectly iff they are mutually orthogonal~\cite{wootters, no-broadcast96}. Specifically, consider a set of $N$ commuting observables $(X^1,X^2,\ldots,X^N)$ in a $d$-dimensional Hilbert space, and let their common eigenbasis be denoted $\cB\equiv \{\ket{i}\}$. Consider a symmetric QCM defined with respect to this common eigenbasis $\cB$ with the parameters $p=1$, $q=0$. As discussed in Sec.~\ref{sec:sym_cloning} above, such a QCM achieves the maximum possible fidelity $\cF_{\rm avg}(\cS, \cB) = 1$ for the ensemble of common eigenstates $\cS$, implying that $\cQ_{c}(\cX) = 0$. 

The converse statement simply follows from the no-cloning principle.  
\end{proof}

We next derive a simple expression for the incompatibility of a general set of quantum observables. 
\begin{lemma}\label{lem:Qc_gen}
Consider a set $\cX$ of $N$ observables on a $d$-dimensional quantum system and let $\cS$ denote the corresponding ensemble of eigenstates. The incompatibility $\cQ_{c} (\cX)$ can be evaluated as,
\begin{equation}
\cQ_{c}(\cX) = 1 - \cF_{\rm avg} (\cS,\cB_{\rm opt}, p_{\rm opt}, q_{\rm opt} ),  \label{eq:Qc_gen}
\end{equation}
where, $\cB_{\rm opt}$ is the optimal cloning basis that maximises the function $A(\cS,\cB)$ defined in Eq.~\eqref{eq:A_B2}. $q_{\rm opt}$ characterizes the optimal quantum cloner and is evaluated as,
\begin{equation}
q_{\rm opt} = \frac{1}{2\sqrt{d-1}}\sqrt{1 - \frac{{\rm sgn} \left(\frac{A_{\rm opt}}{M} -\frac{1}{2}\right)}{\sqrt{1 + (\cG(M,d))^{2}}}}, \label{eq:qopt}
\end{equation}
where $M=Nd$ ${\rm sgn}(x)$ is the signum function defined as 
\[{\rm sgn}(x) = \bigg\lbrace \begin{array}{cc} 
-1 & {\rm if} \; x < 0, \\
0 & {\rm if} \; x = 0, \\
+1 & {\rm if} \; x > 0. 
\end{array} \]
The function $\cG(\cS, N,d)$ is defined as 
\begin{equation}
\cG(\cS,N,d) = \frac{4(A_{\rm opt}(\cS) - Nd)}{(Nd-2A_{\rm opt}(\cS))\sqrt{2(d-1)}}, \label{eq:G_main}
\end{equation}
and $A_{\rm opt}(\cS) = A(\cS, \cB_{\rm opt}) \equiv \max_{\cB}A(\cS,\cB)$ is the maximum value of the quantity $A(\cS,\cB)$, attained for the optimal cloning basis $\cB_{\rm opt}$.
\end{lemma}
Since $p$ and $q$ are related via Eq.~\eqref{eq:unitarity}, knowing the optimal value $q_{\rm opt}$ gives us $p_{\rm opt}$ as well. We refer to Appendix~\ref{sec:Qc_gen} for the detailed calculation leading to the above expression. 

\subsection{Optimal cloning fidelity and incompatibility of MUBs}\label{sec:fopt_Qc_MUB}

We next use the general expression for the optimal cloning fidelity derived above, to find the optimal QCM of a set of mutually unbiased observables. Let $\cX_{\rm MUB} = \{X^{1}, X^{2}, \ldots, X^{N}\}$ denote a set of $N\leq d+1$ MUBs  in $d$-dimensions. Let $\cS_{\rm MUB}$ denote the corresponding eigenstate ensemble, where every state belonging to every basis is picked with a uniform probability of $\frac{1}{Nd}$. 
\begin{lemma}[$\cF_{\rm opt}$ for MUBs] \label{lem:Qc_MUB} The optimal cloning fidelity for the eigenstate ensemble $\cS_{MUB}$ of a set of $N\leq d+1$ MUBs in $d$-dimensions can be evaluated as,
\begin{equation}
\cF_{\rm opt}(\cS_{\rm MUB}) =\cF_{\rm avg}\left(\cS_{\rm MUB}, \cB_{\rm opt} = X^{i}, q_{\rm opt} \right)  ,
\end{equation}
where, $q_{\rm opt}$ is evaluated as in Eq.~\eqref{eq:qopt} above, with $A_{\rm opt} (\cS_{\rm MUB}) = N+d -1$, and the optimal cloning basis $\cB_{\rm opt} = X^{i}$ being any of the bases in the set $\cX$. 
\end{lemma}
Once again, we merely state the result here and refer to Appendix~\ref{sec:Qc_mub} for the details of the calculation.
Our solution for the optimal cloning fidelity for a set of $N\leq d+1$ MUBs in $d$-dimensions constitutes a key step in obtaining an upper bound on the cloning-based incompatibility measure, as shown in Sec.~\ref{sec:Qc_ub} below.  Before proceeding to prove an upper bound on $\cQ_{c}$ for a general set of observables, we would like to note an interesting corollary that arises from our calculations.

\begin{corollary}\label{cor:fopt_d+1}
The optimal quantum cloner for a full set of $N=d+1$ MUBs in $d$-dimensions, whenever it exists, is the universal, symmetric cloner originally introduced in~\cite{buzek_hillery}, achieving an optimal cloning fidelity equal to the universal cloning fidelity calculated in Eq.~\eqref{eq:fopt_univ}.
\end{corollary}
\begin{proof}
The result simply follows from the fact that $A_{\rm opt}(\cS_{\rm MUB}) = N+d -1$, for a set of $N \leq d+1$ MUBs in $d$-dimensions. Furthermore, as shown in Appendix~\ref{sec:Qc_mub}, the optimal cloning basis can be picked to be any of the MUBs in the set. For $N=d+1$, we therefore have, $\frac{A_{\rm opt}(\cS_{\rm MUB})}{Nd} = \frac{2}{d+1}$. Substituting this in the solution for $q_{\rm opt}$ described in Eqs.~\eqref{eq:qopt} and~\eqref{eq:G_main} above, we see that $q^{2}_{\rm opt}(\cS_{MUB}) = \frac{1}{2(d+1)}$ for the full set of $N=d+1$ MUBs. Using the unitarity condition in Eq.~\eqref{eq:unitarity}, we get the optimal value of the parameter $p$ as, $p_{\rm opt} = \frac{2}{d+1}$. The optimal QCM thus satisfies $p^{\rm opt} = 2p_{\rm opt}q_{\rm opt}$, which is the same as the universal, symmetric QCM described in~\cite{buzek_hillery} and the optimal cloning fidelity is indeed the universal cloning fidelity calculated in Eq.~\eqref{eq:fopt_univ}. 
\end{proof}

Although the form of $\cF_{\rm opt}(\cS_{\rm MUB})$ is not a simple function to write down, it can of course be computed easily for a given number $N$ and dimension $d$. Knowing the optimal fidelity, we can immediately calculate the mutual incompatibility $\cQ_{c}$ of a set of $N$ MUBs in $d$-dimensions. We have plotted $\cQ_{c} (\cX_{\rm MUB})$ as a function of $d$ and $N$, in Figs.~\ref{fig:Q_d} and~\ref{fig:Q_N} respectively.

\subsection{Upper bound on $\cQ_{c}$}\label{sec:Qc_ub}

We next show that there exists a non-trivial upper bound for the measure $\cQ_{c}(\cX)$, as a function of the number of observables $N$ in the set $\cX$ and the dimension $d$ of the space. Furthermore, whenever $N\leq d+1$, this upper bound is tight for a set of $N$ MUBs, thus reiterating the fact that mutually unbiased observables,whenever they exist, are maximally incompatible. 

\begin{theorem}[Upper bound on $\cQ_{c}$]\label{thm:Qc_ub}
Consider a set $\cX$ of $N$ observables on a $d$-dimensional quantum system. The measure $\cQ_{c} (\cX)$ satisfies,
\begin{eqnarray}
\cQ_{c}(\cX) &\leq& 1 - \cF_{\rm opt}(\cS_{\rm MUB}), \; N \leq d+1 \label{eq:Qc_ub1} \\
\cQ_{c}(\cX) &\leq & 1 - \cF_{\rm univ}(d), \qquad N \geq d +1 . \label{eq:Qc_ub2}
\end{eqnarray}
Here, $\cF_{\rm opt}(\cS_{\rm MUB})$ is the optimal cloning fidelity attained for a set of $N$ MUBs, and $\cF_{\rm univ}(d)$ is the cloning fidelity attained by a universal QCM in $d$-dimensions.

For $N \leq d+1$, equality is attained iff the set of observables in $\cX$ are mutually unbiased and the corresponding optimal QCM is characterised by $q=q_{\rm opt}$. For $N > d+1$, the upper bound is attained for a set of observables whose eigenstates form a unitarily invariant ensemble~\cite{Fuchs-Sasaki}.
\end{theorem}
We should note here that the bounds in Eq.~\eqref{eq:Qc_ub1} and Eq.~\eqref{eq:Qc_ub2} are the same when $N=d+1$. For $N< d+1$, the optimal cloning fidelity for a set of $N$ MUBs ($\cF_{\rm opt}(\cS_{\rm MUB})$) is always greater than the fidelity attained by a universal QCM ($\cF_{\rm univ}(d)$), so  the bound in Eq.~\eqref{eq:Qc_ub1} holds. Since the number of MUBs in $d$-dimensions is bounded by $d+1$, the bound in Eq.~\eqref{eq:Qc_ub2} takes over for $N>d+1$.

\begin{proof}
We first prove the case where $N \leq d+1$. The main intuition behind the upperbound is the fact that for a given QCM with parameters $(p,q)$, the average fidelity attained using the optimal cloning basis for any ensemble $\cS$ is always greater than that attained for an ensemble of MUBs. This is explicitly shown below.

Consider a set of $N \leq d+1$ observables $\cX$, in a $d$-dimensional space. The corresponding eigenstate ensemble $\cS$ has $Nd$ states. Comparing the average cloning fidelity attained for this ensemble using a symmetric QCM in basis $\cB$, with parameters $(p,q)$, with that attained for an ensemble $\cS_{\rm MUB}$ of $N$ MUBs, we have, 
\begin{align}
&\cF_{\rm avg}(\cS,\cB,q)-\cF_{\rm avg}(\cS_{MUB},\cB,q) \nonumber \\
&= \frac{(p^2-2pq)}{Nd} ( \cA(\cS,\cB) -\cA(\cS_{MUB},\cB) ).\nonumber
\end{align}
Recall that $\frac{\cA_{opt}(\cS)}{Nd}$ is simply the same as $\mathbb{F}(\cS, \cM_{\cB}, \cA_{\cB})$, the average fidelity obtained via a projective measurement followed by state reconstruction, as observed in Prop.~\ref{prop:zero_q}. Furthermore, we know that,
\[ \frac{\cA_{opt}(\cS)}{Nd} \geq \frac{N+d-1}{Nd}, \]
where the minimum is attained when the ensemble $\cS$ is in fact $\cS_{\rm MUB}$ as proved in~\cite{incompatibility_BM} and Lemma~\ref{lem:proj_opt}. Finally, since we define the optimal value of the cloning fidelity in the region of the parameter space where $p^{2} > 2pq$, we have,
\begin{align}
& \sup_{\cB}\cF_{\rm avg}(\cS,\cB,q)-\sup_{\cB}\cF_{\rm avg}(\cS_{MUB},\cB,q) \nonumber \\
&= \frac{(\cA_{opt}(\cS)-\cA_{opt}(\cS_{MUB}))(p^2-2pq)}{Nd}
& \geq 0. \label{eq:max_B}
\end{align}

Let $(p_{\rm opt},q_{\rm opt})$ be the optimal value of the parameters $(p,q)$ corresponding to the ensemble $\cS_{MUB}$. Then, the optimal cloning fidelity for an ensemble $\cS$ satisfies
\begin{align}
\cF_{\rm opt}(\cS)  & = \sup_{\cB,p,q}\cF_{\rm avg}(\cS,\cB,p,q) \nonumber \\
& \geq \sup_{\cB} \cF_{\rm avg}(\cS,\cB,p_{\rm opt},q_{\rm opt}) \nonumber \\
& \geq \sup_{\cB}\cF(\cS_{MUB},\cB,p_{\rm opt},q_{\rm opt}) \nonumber \\
& = \cF_{\rm opt}(\cS_{\rm MUB}),  \label{eq:fopt_lb}
\end{align}
where the last inequality follows from Eq.~\eqref{eq:max_B}.

The upper bound now simply follows from the definition of the  incompatibility measure $\cQ_{c}$ -- for any set of observables $\cX$, 
\begin{equation}
\cQ_c(\cX) = 1 - \cF_{\rm opt} (\cS) \leq 1 - \cF_{\rm opt}(\cS_{\rm MUB}). 
\end{equation}
It is easy to see that the sequence of inequalities collapses when the set $\cX \equiv \cX_{\rm MUB}$ is a set of MUB observables, showing that the bound is tight for $\cX_{\rm MUB}$.

In the case where $N> d+1$, we use a general bound on the optimal fidelity achieved via the measurement-and-reconstruction protocol described in Appendix~\ref{sec:Q_avgFid}. For any ensemble of states $\cS$ in a $d$-dimensional Hilbert space, the optimal fidelity $\mathbb{F}_{\rm opt}(\cS)$ (defined in Eq.~\eqref{eq: f_max}) satisfies~\cite{Fuchs-Sasaki},
\[ \mathbb{F}_{\rm opt}(\cS) \geq \frac{2}{d+1},\]
with the optimal measurement being a projective measurement in a random basis. Therefore, we see from the equality in Eq.~\eqref{eq:cloning_fid_equality} that, 
\begin{equation}
\frac{A_{\rm opt}(\cS)}{Nd} \equiv \max_{\cB}\mathbb{F}_{\rm avg} (\cS, \cM_{\cB}, \cA_{\cB}) \geq \frac{2}{d+1}.
\end{equation}
Recalling the definitions in Eqs.~\eqref{eq:avg_fid2}~\eqref{eq:opt_cloning_fid}, we can evaluate the optimal cloning fidelity for any ensemble of $Nd$ states in $d$-dimensions is given by,
\begin{eqnarray}
\cF_{\rm opt}(\cS) &\geq& \max_{q}\frac{1}{d+1}\bigg[ 2 + (d-1)(d-3)q^{2} \nonumber \\
&+&  2(d-1)q\sqrt{1 - 2(d-1)q^2} \bigg] .
\end{eqnarray} 
Solving this optimization problem, we get,
\begin{equation} 
q_{\rm opt}(\cS) = \sqrt{\frac{1}{2(d+1)}}, \label{eq:qopt_d+1}
\end{equation}
which matches the parameter $q$ for a universal, symmetric cloner in $d$-dimensions~\cite{buzek_hillery}. As shown in Appendix~\ref{sec:N_d+1}, this implies that the optimal cloning fidelity is bounded from below as,
\begin{equation}
\cF_{\rm opt}(\cS) \geq \frac{d-3}{2(d+1)}. \label{eq:fopt_d+1}
\end{equation}
This gives the desired upper bound for a set of $N > d+1$ observables. The tightness of the bound follows from the construction of the optimal ensemble attaining the bound in~\cite{Fuchs-Sasaki}.
\end{proof}

\begin{center}
\begin{figure}[t]
\centering
\includegraphics[scale=0.48]{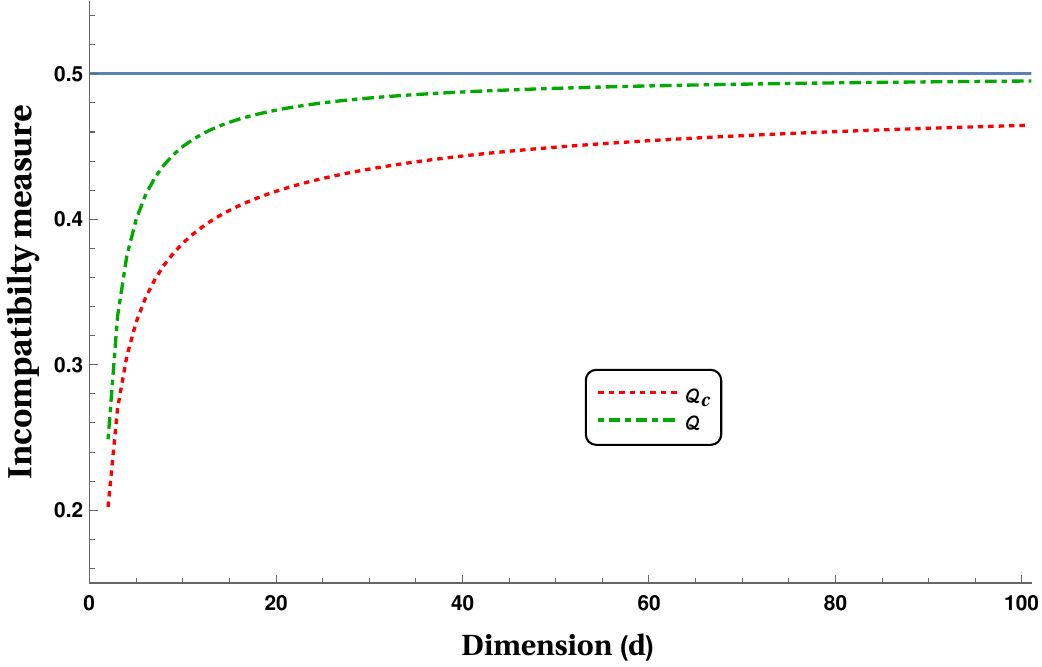}
\caption{Incompatibility of a pair of MUBs as a function of dimension $d$. $\cQ_{c}$ denotes the cloning-based incompatibility measure and $\cQ$ denotes the measure based on measurement-and-reconstruction.}
\label{fig:Q_d}
\end{figure}
\end{center}

To gain further insight into the behaviour of the cloning-based measure $\cQ_{c}$, we plot the upper bound proved in Eq.~\eqref{eq:Qc_ub1}, as function of the number of observables $N$ as well as the dimension $d$. Fig.~\ref{fig:Q_d} plots the mutual incompatibility, as quantified by the measure $\cQ_{c}$, of a pair of MUBs ($N=2$) as a function of the dimension $d$.  For comparison, we have also plotted the corresponding upper bound on the measure $\cQ$ which comes from a measurement-and-reconstruction strategy, as discussed in Appendix.~\ref{sec:Q_avgFid}. 

Fig.~\ref{fig:Q_N} shows how the incompatibility $\cQ_{c}$ grows as a function of the number of MUBs $N$, while keeping the system dimension fixed ($d=11$). Once again we have plotted the measurement-and-reconstruction based measure $\cQ$ for comparison. 

In both cases, the upper bound on the cloning-based incompatibility measure $\cQ_{c}$ is always lower than the bound on the measure $\cQ$. In other words, the optimal cloning fidelity for a set of MUBs is always higher than the best fidelity achievable via a measurement-and-reconstruction strategy. This is to be expected since the optimal fidelity in a measurement-and-reconstruction protocol for a set of MUBs  is always achieved by a projective measurement. And Prop.~\ref{prop:zero_q} implies that whenever the optimal measurement is a projective measurement, the measurement-and-reconstruction strategy is a special case of the symmetric QCM with parameters $p=1,q=0$. The optimal cloning fidelity, on the other hand is calculated by optimizing over all values of $(p,q)$ and is therefore always higher, leading to the fact that $\cQ_{c}$ is always lower than $\cQ$ for a set of $N$ MUBs in $d$ dimensions. In the context of quantum key distribution (QKD) this reiterates the fact that an optimal cloning attack would give the eavesdropper a better fidelity than an intercept-and-resend attack, when the signal states are drawn from a uniform ensemble of MUBs.

\subsection{Example: the case of two qubit observables}\label{sec:qubit}

\begin{center}
\begin{figure}[t!]
\includegraphics[scale=0.48]{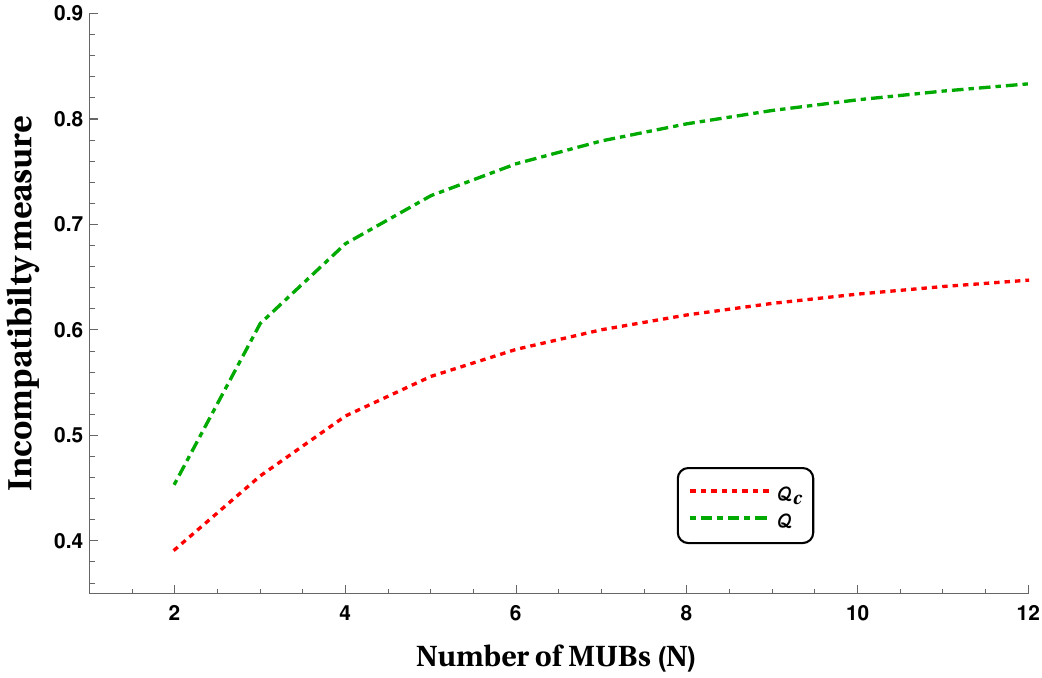}
\caption{Incompatibility of a set of $N$ MUBs in $d=11$ dimensions. $\cQ_{c}$ denotes the cloning-based incompatibility measure and $\cQ$ denotes the measure based on measurement-and-reconstruction.}
\label{fig:Q_N}
\end{figure}
\end{center}

We now present a simple, yet insightful example of how our results maybe applied to calculate the incompatibility of an arbitrary set of quantum observables. Specifically, we consider a pair of spin-$\frac{1}{2}$ observables, $A$ and $B$, which maybe parametrized in terms of unit vectors $\vec{a},\vec{b} \in \mathbb{R}^{3}$ on the Bloch sphere, as,
\[A = \alpha_{1} I + \alpha_{2}\vec{a}.\vec{\sigma}, \; B =  \beta_{1} I + \beta_{2}\vec{b}.\vec{\sigma},\]
where $\{\alpha_{i},\beta_{i}\}$ are real parameters and $\vec{\sigma} \equiv (\sigma_{X}, \sigma_{Y},\sigma_{Z})$ denotes the vector of Pauli operators. The mutual incompatibility of such a pair of qubit observables can be quantified using our formalism, as follows. 

Let $\ket{a_{\pm}}$ and $\ket{b_{\pm}}$ denote the eigenstates of the spin-$\frac{1}{2}$ observables $A$ and $B$ respectively. Then,
\begin{eqnarray}
\ket{a_{\pm}}\bra{a_{\pm}} &=& \frac{\Id\pm\vec{a}.\vec{\sigma}}{2} \nonumber \\
\ket{b_{\pm}}\bra{b_{\pm}} &=& \frac{\Id\pm\vec{b}.\vec{\sigma}}{2} . \nonumber
\end{eqnarray}
The first step in calculating the incompatibility $\cQ_{c}$ is to compute the optimum over all basis choices $\cB$, of the quantity $A(\cS_{2},\cB)$, where $\cS_{2}$ is the eigenstate ensemble,
\begin{equation}
\cS_{2} = \left\{ \frac{1}{4}, \ket{a_{\pm}}\bra{a_{\pm}}, \ket{b_{\pm}}\bra{b_{\pm}}\right\}. \label{eq:S_qubit}
\end{equation}
Once we solve for $A_{\rm opt}(\cS_{2}) = \max_{\cB}A_{\rm opt} (\cS_{2}, \cB)$, we can use the expressions from Lemma~\ref{lem:Qc_gen} to obtain the average cloning fidelity and hence the incompatibility $\cQ_{c}$. The optimal QCM for a pair of qubit observables is described in the following lemma. 

\begin{lemma}\label{lem:qubit}
For the ensemble $\cS_{2}$ defined in Eq.~\eqref{eq:S_qubit}, the optimal cloner has parameters $(p_{\rm opt}, q_{\rm opt})$ where,
\begin{eqnarray}
q_{\rm opt} &=& \frac{1}{2}\sqrt{1 - \frac{1}{\sqrt{1 + (\cG(\vert\vec{a}.\vec{b}\vert))^{2}}}}, \nonumber \\
\cG(\vert\vec{a}.\vec{b}\vert) &=& \frac{\sqrt{2}(1-\vert\vec{a}.\vec{b}\vert)}{1 + \vert\vec{a}.\vec{b}\vert}. \label{eq:qopt_qubit}
\end{eqnarray}
The optimal cloning bases are characterized by the Bloch vectors 
\begin{equation}
\vec{r}_{\pm} = \frac{\vec{a} \pm \vec{b}}{\vert \vec{a} \pm \vec{b}\vert} . \label{ropt_qubit}
\end{equation}

\end{lemma}


\begin{proof}

Consider a projective measurement along an arbitrary direction $\vec{r}$ in the Bloch sphere. The measurement basis $\cB$ corresponding to such a measurement can be described by a pair of orthonormal vectors $|\psi_{\pm}\rangle$ given by,
\begin{equation}
\ket{\psi_{\pm}}\bra{\psi_{\pm}}=\frac{\Id \pm \vec{r}.\vec{\sigma}}{2} .
\end{equation}
The probabilities of obtaining outcomes $\pm$ are given by, 
\begin{eqnarray}
 p(+) _{\ket{a_{\pm}}} &=& \left[\frac{1\pm\vec{r}.\vec{a}}{2}\right], \; p(-)_{\ket{a_{\pm}}} = \left[\frac{1\mp\vec{r}.\vec{a}}{2}\right] \nonumber \\
 p(+)_{\ket{b_{\pm}}} &=& \left[\frac{1\pm\vec{r}.\vec{b}}{2}\right], \;  p(-)_{\ket{b_{\pm}}} \left[\frac{1\mp\vec{r}.\vec{b}}{2}\right]. \nonumber 
\end{eqnarray}
Therefore, the quantity $A (\cS_{2},\cB)$, is given by,
\begin{align}
A(\cS_{2},\cB) &=  p^{2}(+)_{\ket{a_{\pm}}} + p^{2}(-)_{\ket{a_{\pm}}} \nonumber \\
&+  p^{2}(+)_{\ket{b_{\pm}}} + p^{2}(-)_{\ket{b_{\pm}}} \nonumber\\
&= 2\left[\frac{1+\vec{r}.\vec{a}}{2}\right]^2+2\left[\frac{1-\vec{r}.\vec{a}}{2}\right]^2 \nonumber \\
&+ 2\left[\frac{1+\vec{r}.\vec{b}}{2}\right]^2 + 2\left[\frac{1-\vec{r}.\vec{b}}{2}\right]^2 \nonumber\\
&= 2+ (\vec{r}.\vec{a})^2+ (\vec{r}.\vec{b})^2 .
\end{align}
The optimal cloning basis is characterized by the vector $\vec{r}$ for which $A(\cS_{2},\cB)$ attains its maximum value. Therefore,
\[
A_{\rm opt}(\cS_{2}) = \max_{\vec{r}} \left[ 2+ (\vec{r}.\vec{a})^2+ (\vec{r}.\vec{b})^2  . \right] \]
Following earlier works relating to entropic uncertainty bounds for qubit observables~\cite{ghirardi, bosyk, PM_MS}, we first argue that the maximum is attained when the vector $\vec{r}$ is coplanar with $\vec{a}$ and $\vec{b}$~\footnote{Given any vector $\vec{v}_{\perp}$ in a plane perpendicular to the plane containing $\vec{a}$ and $\vec{b}$, we can always find a corresponding vector $\vec{v}_{c}$ in the intersection of the two planes such that $\vert\vec{v}_{c}.\vec{a}\vert \geq \vert \vec{v}_{\perp}. \vec{a}\vert $. Since the function $x^{2}$ is monotonically increasing for $x >0$, we can restrict ourselves to maximizing over vectors in the plane containing $\vec{a}$ and $\vec{b}$}.

Let the angle between $\vec{a}$ and $\vec{r}$ be $\theta$ and let $\gamma$ denote the angle between $\vec{a}$ and $\vec{b}$. Then, $\cos\theta = \vec{a}.\vec{r}$, $\cos\gamma = \vec{a}.\vec{b}$, and since $\vec{a}, \vec{b}$ and $\vec{r}$ are coplanar, $\vec{r}.\vec{b} = \cos(\theta-\gamma)$. Then, 
\[A_{\rm opt}(\cS_{2},\cB) = \max_{\theta} \left[ 2 + \cos^{2}\theta + \cos^2(\theta-\gamma) \right].  \]
Taking the derivative with respect to $\theta$, we see that the extremal values are attained for $\tan2\theta_{\rm opt} = \tan\gamma$, implying that $\theta_{\rm opt} = \frac{\gamma}{2} + n\frac{\pi}{2}$. Checking the second derivative, we see that $A_{\rm opt}$ is maximized for $\theta_{\rm opt} = \frac{\gamma}{2}$ when $0\leq \gamma < \frac{\pi}{2}$ and $\theta_{\rm opt} = \frac{\gamma}{2} + \frac{\pi}{2}$ when $\frac{\pi}{2} < \gamma \leq \pi$. For $\gamma=\frac{\pi}{2}$,  $A(\cS_{2},\cB)=3$, which is independent of $\theta$ and so $A(\cS_{2},\cB)$ is optimal for any basis $\cB$ corresponding to a vector $\vec{r}$ coplanar with $\vec{a}$ and $\vec{b}$.  Therefore, 
\begin{equation}
A_{\rm opt} (\cS_{2}) = \bigg\lbrace \begin{array} {cc} 2\left( 1 + \cos^{2}\frac{\gamma}{2}\right), & 0\leq\gamma< \frac{\pi}{2},  \\
2\left(1 +\sin^{2}\frac{\gamma}{2} \right), & \frac{\pi}{2} < \gamma \leq \pi ,\\
3 , & \gamma=\frac{\pi}{2}. 
\end{array}
\end{equation}
In terms of the vectors $\vec{a}, \vec{b}$ characterizing the ensemble $\cS_{2}$,  we have,
\begin{equation}
A_{\rm opt}(\cS_{2}) = 3 + \vert \vec{a}.\vec{b}\vert. \label{eq:A_qubit}
\end{equation}
For $\gamma\neq\frac{\pi}{2}$, corresponding to the two optimal values $\theta_{\rm opt}=\frac{\gamma}{2}, \frac{\gamma}{2} + \frac{\pi}{2}$, we see that the Bloch vectors $\vec{r}_{\pm}$ corresponding to the optimal cloning bases are, 
\begin{equation}
\vec{r}_{\pm} = \frac{\vec{a} \pm \vec{b}}{\vert \vec{a} \pm \vec{b}\vert} .
\end{equation}
Using the value of $A_{\rm opt}(\cS_{2})$ in Eqs.~\eqref{eq:qopt} and~\eqref{eq:G_main}, we get the desired parameters of the optimal QCM.
\end{proof}

Having characterized the optimal QCM for a pair of qubit observables, we can then obtain an expression for the average cloning fidelity $\cF_{\rm avg}(\cS_{2})$ and the incompatibility measure $\cQ_{c}(\cS_{2})$ as a function of the inner product $\vec{a}.\vec{b}$. These expressions are evaluated in Appendix~\ref{sec:qubit_app}. The incompatibility measure is indeed monotonic in the inner product $\vec{a}.\vec{b}$, as shown in the plot in figure \eqref{fig:Q_N}.

\begin{center}
\begin{figure}[t!]
\includegraphics[scale=0.54]{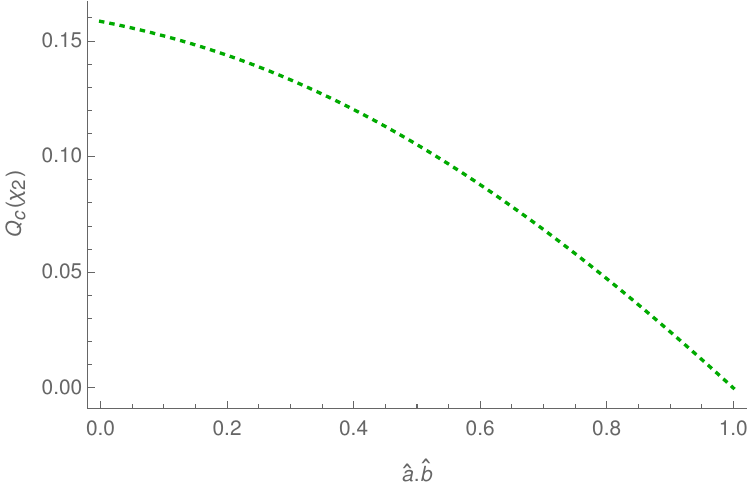}
\caption{Incompatibility of a pair of dichotomic qubit observables ($N=2$, $d=2$), as a function of the overlap of their Bloch vectors ($\vec{a}.\vec{b}$).}
\label{fig:Q_N}
\end{figure}
\end{center}

A quick check reveals that our solution for the optimal QCM for a pair of dichotomic observables matches with the known limiting cases. When $\vec{a}.\vec{b} = 0$, the observables $A$ and $B$ are mutually unbiased and from Eq.~\eqref{eq:A_qubit}, we get, 
\[ \frac{A_{\rm opt}(\cS_{2})}{Nd} = \frac{3}{4} = \frac{N + d-1}{Nd},\] with $N=2$, $d=2$, as expected. When $\vert\vec{a}.\vec{b}\vert = 1$, the observables $A$ and $B$ commute. Our solution gives$\frac{A_{\rm opt}(\cS_{2})}{Nd} = 1$, which implies that the states in the ensemble can be discriminated perfectly using a projective measurement. Furthermore, the optimal cloner has $q_{\rm opt} = 0$ (since $\cG = 0$ in Eq.~\eqref{eq:qopt_qubit}) and $p_{\rm opt} = 1$, as expected from Lemma~\ref{lem:Qc_lb}. 

\section{Quantifying incompatibility of general quantum measurements}\label{sec:general}

Finally, we show that our cloning-based incompatibility measure can be generalized beyond the case of rank-one projective measurements (PVMs) to quantify the incompatibility of a set of general quantum measurements. Recall that the measure $\cQ_{c}$ is based on associating a set of quantum observables $\cX$ with a set of nondegenerate Hermitian operators, so that the corresponding ensemble of eigenstates $\cS$, simply becomes a set of rank-one projectors. Consider now, a set of degenerate observables $\cX = \{X_1,....,X_N\}$ on a Hilbert space $\cH$ of dimension $d$, such that the Hermitian operator $X_i$ has $m_i$ degenerate eigenspaces,  for all $i=1,\ldots,N$. 

We now define two ensembles associated with such a set $\cX$. Let $\cS_{\rm nd}$ denote the collection of eigenstates of the nondegenerate subspaces of all the observables in $\cX$, wherein each eigenstate is picked with a uniform probability of $\frac{1}{Nd}$. Similarly, let $\cS_{\rm d}$ denote the collection of eigenstates of the degenerate subspaces of all the observables in $\chi$, where the eigenstates are each picked with uniform probability $\frac{1}{Nd}$. Thus the complete eigenstate ensemble associated with a set of degenerate observables $\cX$ is given by $\cS = \cS_{\rm d}\cup \cS_{\rm nd}$. Clearly, $\cS_{\rm d}$ can be changed by choosing different sets of eigenstates spanning the degenerate subspaces and therefore $\cS$ can be changed by changing $\cS_{\rm d}$. Let $\mu$ be the set of all such unitarily equivalent eigenstate ensembles. So, for a set of degenerate observables, the modified measure of incompatibility is,
\begin{equation}
\cQ_{c}(\cX)=1 - \max_{\cS \in \mu}\cF_{\rm opt}(\cS).
\end{equation}
Clearly, $\cQ_{c}(\chi)=0$ for compatible set and $\cQ_{c}(\cX) > 0$ otherwise.


We may further generalise our measure to a a set of general quantum measurements described by a set of positive operator valued measures (POVMs), by invoking the Naimark dilation theorem. The Naimark extension has been used to characterize the incompatibility of a pair of POVMs~\cite{beneduci, haapasalo} in the past. Recently, in independent work, we have shown that any set of $N$ POVMs is compatible iff there exists a common Naimark extension of this set such that the corresponding  PVMs are pairwise commuting~\cite{naimark}. Let, $\cX=\{E_1,E_2,\ldots,E_N\}$ be a set of POVMs and $\Pi=\{P_1,\ldots,P_N\}$ be a corresponding set of PVMs constructed via a common Naimark extension (see~\cite{naimark} for the details). Since such a Naimark extension is not unique, one can indeed have different sets $\Pi$ corresponding to the set $\cX$. Let, $\nu$ be the set of all such Naimark extensions of $\cX$, obtained using a common ancilla space and a common ancilla state. Then we may define the incompatibility of the set $\cX$ as,

\begin{equation}
\cQ_{c}(\cX)=\max_{\Pi\in\nu}\cQ_{c}(\Pi).
\end{equation}

It follows from the central result proved in~\cite{naimark}, that $\cQ(\cX)=0$ for a compatible set and $\cQ(\cX) > 0$ otherwise.

\section{Conclusions}\label{sec:concl}

In summary, we have proposed a new approach to quantifying incompatibility based on symmetric quantum cloners. We establish a quantitative relation between the incompatibility of a set of quantum observables and the optimal fidelity with which their eigenstate ensembles can be cloned using a symmetric QCM. We show that our measure satisfies the desirable properties which any faithful measure of incompatibility should, namely, it vanishes only for a set of commuting observables and attains its maximum value for a set of mutually unbiased bases. 

Our investigation brings to light an interesting facet of the optimal cloning fidelity, namely, that the optimal cloning basis for a given ensemble of states is in fact the projective measurement that achieves the best fidelity in a measurement-and-reconstruction protocol.  We use this connection to fully characterize the optimal cloning machine corresponding to an arbitrary ensemble of quantum states. For example, we show that the optimal cloner for the complete set of $d+1$ MUBs in $d$-dimensions is simply the universal, symmetric QCM discussed in~\cite{buzek_hillery}. 

From a foundational point of view, the cloning-based incompatibility measure discussed here thus firms up a connection which has long been intuited, between two fundamental principles in quantum information theory, namely, the no-cloning principle and the existence of mutually incompatible observables. On the practical side, our results become relevant in the context of QKD, where it is important to know the optimal cloning fidelity that an eavesdropper may get, for a given ensemble of signal states prepared by the sender. Finally, we note that the optimal cloning fidelity provides an alternative characterization of the {\emph quantumness} of an ensemble, an idea which was originally defined using a measurement-and-reconstruction protocol~\cite{Fuchs-Sasaki}. An interesting direction for future work is to study how the loss of coherence~\cite{coherence} in the input ensemble impacts the optimal cloning fidelity. 

Our analysis shows that for those ensembles for which the optimal fidelity in a measurement-and-reconstruction protocol is achieved for a projective measurement, the measurement-and-reconstruction strategy is a special case of the $1\rightarrow 2 $ symmetric QCM. An interesting open question in this context is to ask whether a similar connection can be made for a more general measurement-and-reconstruction strategy, involving POVMs. This might require us to expand the definition of optimal cloning fidelity beyond the symmetric QCMs considered here.

In further work, we would like to extend our cloning-based approach to obtain a quantitative estimate of the incompatibility of certain exemplary sets of positive operator value measures (POVMs). A first step in this direction has already been made in our recent work on characterizing the incompatibility of POVMs via their Naimark extensions~\cite{naimark} into PVMs. The key question to be addressed in this context, is to identify a set of characteristic Naimark extensions which can be used to witness the incompatibility of a set of POVMs.


\section{Acknowledgements}
The authors would like to thank Prof. Sibasish Ghosh for several valuable discussions and inputs. PM is grateful to Prof. S. Lakshmibala for insightful early discussions. PM acknowledges financial support by the Department of Science and Technology, Govt. of India, under grant number DST/ICPS/QuST/Theme-3/2019/Q69.

\appendix

\section{Quantifying incomptibility via a measurement-and-reconstruction protocol}\label{sec:Q_avgFid}

We briefly discuss a recently proposed operational measure of incompatibility~\cite{incompatibility_BM}. The measure $\cQ$ defined in~\cite{incompatibility_BM} is based on a measurement and reconstruction protocol, somewhat reminiscent of a na\"ive cloning protocol, naturally leading to the question of how it would compare to a more general cloning strategy using, for example, a symmetric quantum cloner. Furthermore, in a departure from the more widely studied measures of incompatibility based on uncertainty relations, the $\cQ$ measure does not vanish for observables that commute over a subspace; it is zero only for observables that commute over the entire space.

Consider a set $\Pi=\{\Pi_1,\Pi_2,......\Pi_N\}$ of $N$ observables on a $d$-dimensional Hilbert space $\mathcal{H}$, and let $\Pi^j_i=\ket{\psi^j_i}\bra{\psi^j_i}$ denote the $j^{\rm th}$ eigenstate of the $i^{\rm th}$ observable $\Pi_i$. We may note here that our measure assumes the canonical association of a physical observable with a Hermitian operator. The measure $\cQ(\Pi)$ aims to quantify the incompatibility of the set $\Pi$ based on the non-orthogonality of the corresponding ensemble of eigenstates, as follows. 

Let $\cS(\Pi) \equiv \{\Pi^{j}_{i}, \, 1\leq i\leq N, \, 1\leq j\leq d\}$ denote the ensemble of eigenstates of the observables in the set $\Pi$. Consider a two-party protocol wherein the sender Alice draws states uniformly at random from the set $\cS(\Pi)$ so that the probability of a given state is $\frac{1}{Nd}$, and transmits them. The receiver Bob makes some generalised measurement $\cM = \{M_a\}$ and adopts a reconstruction strategy $\cA$, whereby, he reconstructs the state $\sigma_a$ when he gets outcome $a$. Now the \emph{average fidelity} between the final states obtained by Bob via this measurement and reconstruction strategy, and the initial set of states sent by Alice is given by,

\begin{equation}
\mathbb{F}_{\rm avg}(\cS, \cM,\cA)=\frac{1}{Nd}\,\sum_{ija}\tr(\Pi^i_j M_a)\tr(\Pi^i_j\sigma_a) . \label{eq:fid_avg1}
\end{equation}

Maximizing over all measurements and state-reconstructions, we get the \emph{optimal} fidelity for the ensemble $\cS$ as,
\begin{equation}
\mathbb{F}_{\rm opt}(\cS) = \sup_{\cM} \sup_{\cA} \mathbb{F}_{\rm avg}( \cS, \cM, \cA) . \label{eq: f_max}
\end{equation}

We then define the measure $\cQ$ for the set $\Pi$ as,
\begin{equation}
\cQ ({\Pi}) = 1 - \mathbb{F}_{\rm opt}(\cS) . \label{eq: Q_measure}
\end{equation}

Clearly, for commuting observables, which share a common set of eigenstates, we can achieve $\mathbb{F}_{\rm opt}=1$ by making the choice $M_a,\sigma_a = \ket{\psi^a_i}\bra{\psi^a_i}$, where $|\psi^{a}_{i}\rangle$ is the $a^{\rm th}$ eigenstate of the $i^{\rm th}$ observable in the set $\Pi$. This in turn implies that $\cQ=0$ for a commuting set of observables. For any other set of observables $Q > 0$, as shown in~\cite{incompatibility_BM}, making this a true measure of incompatibility of any set of observables $\Pi$. It is further proved that for a set of $N$ observables in a $d$ dimensional Hilbert space,
\begin{equation}
\mathbb{F}_{\rm opt}(\cS) \geq \frac{N+d-1}{Nd} , \label{eq:F_bound}
\end{equation}
so that,
\begin{equation}
0 \leq \cQ(\Pi) \leq \left( 1- \frac{1}{N}\right)\left(1 - \frac{1}{d}\right).\label{eq:Q_ub}
\end{equation}
The upper bound is shown to be achieved for a set of mutually unbiased observables, namely, a set of observables $\{\Pi_{i}\}$ with the property, 
\begin{eqnarray}
\tr(\Pi^{i}_{j}\Pi^{k}_{j}) &=& \delta_{ik}, \, \forall j = 1,2,\ldots, N. \nonumber \\
\tr( \Pi^{i}_{j}\Pi^{k}_{l}) &=& \frac{1}{d}, \, \forall j \neq l , \, \forall \, i,k = 1,2,\ldots,d. \label{eq:mub}
\end{eqnarray}
While the optimal measurement and reconstruction strategy that attains $\mathbb{F}_{\rm opt}$ is not known in general, for such mutually unbiased bases (MUBs), we show that projective measurements are optimal in the sense that they achieve the lower bound given in Eq.~\eqref{eq:F_bound}. The formal statement and proof of this fact are given below.

\subsection{Optimal measurement and reconstruction strategy for MUBs}\label{sec:Fopt_MUB}

\begin{lemma}\label{lem:proj_opt}
For an ensemble of states comprising vectors from a set $\Pi$ of $N$ mutually unbiased bases (MUBs), the optimal fidelity $\mathbb{F}_{\rm opt}$ defined in Eq.~\eqref{eq: f_max} is achieved by a projective measurement, with the measurement basis chosen to be any of the bases in the set $\cS$ where $\cS$ is the set of eigenstates of the operators in $\Pi$ .
\end{lemma}

\begin{proof}
Recall that for a pair of mutually unbiased bases, $\Pi_{j}, \Pi_{l}$,
$\tr(\Pi^i_j\Pi^k_j)=\delta_{ik}$, and for $ j\neq l$, $\tr(\Pi^i_j\Pi^k_l)=\frac{1}{d}$. Define a measurement $\cM$ with the rank-one projectors $\{M_a = \Pi^a_1 \}$ which are simply the basis vectors of the first basis $\Pi_{1}$ in the mutually unbiased set. Similarly we choose the reconstruction strategy given by $\sigma_a = \Pi^{a}_{1}$. Then the average fidelity can be computed as,
\begin{align}
\mathbb{F}(\cM,\cA)&=\frac{1}{Nd}\left[\sum_{i,a=1}^{d}\left(\tr(\Pi^i_1\Pi^a_1)\right)^2 + \sum_{i,a=1}^{d}\sum_{j\neq 1} \left( \tr(\Pi^i_j\Pi^a_1)\right)^2\right]\nonumber\\
&=\frac{1}{Nd}\left[\sum_{i,a=1}^{d}\delta_{ia}^2+\sum_{i,a}\sum_{j\neq 1}\frac{1}{d^2} \right]\nonumber\\
&=\frac{1}{Nd}\left[ d+\frac{(N-1)d^2}{d^2}\right]\nonumber\\
&=\frac{N+d-1}{Nd} . \nonumber
\end{align}

Comparing with Eqs.~\eqref{eq: f_max},~\eqref{eq:Q_ub}, and using the fact that these bounds are tight for a set of $N$ MUBs in $d$-dimensions (proved in~\cite{incompatibility_BM}), we see that the above choice of projective measurement $\cM$ and reconstruction strategy $\cA$ does achieve the value of $\mathbb{F}_{\rm opt}(\cS_{\rm MUB})$.
\end{proof}

\section{Optimal cloning fidelity for an ensemble of states}\label{sec:Qc_gen}

Recall that the optimal cloning fidelity for an ensemble $\cS$ of $M$ quantum states in $d$-dimensions is defined as,
\[ \cF_{\rm opt}(\cS)  = \underset{p,q: p^{2} > 2pq}{\rm max}\max_{\cB} \cF_{\rm avg}(\cS,\cB,p,q). \]

From the form of the average fidelity written down in Eq.~\eqref{eq:avg_fid2}, we see that,
\begin{eqnarray}
\cF_{\rm avg}(\cS, \cB, p, q) &=& 2pq +  (d-1)q^{2} \nonumber \\
&+& \left(\frac{A(\cS, \cB)}{M} \right)(p^{2}-2pq). \nonumber
\end{eqnarray}
The average cloning fidelity $\cF_{\rm avg}(\cS,\cB, p,q)$ depends on cloning basis $\cB$ only via the term $A(\cS, \cB)/M$. Furthermore, from Prop.~\ref{prop:zero_q} it follows that,
\[ \max_{\cB} \frac{A(\cS, \cB)}{M} \equiv \max_{\cB}\mathbb{F}_{\rm avg}(\cS, \cM_{\cB}, \cA_{\cB}),\]
where the RHS is the maximum average fidelity attained via a measure-and-reconstruct strategy, as defined in Sec.~\ref{sec:Q_avgFid}. 

Let $A_{\rm opt}(\cS) = \max_{\cB} A(\cS, \cB)$ denote the optimum value for the ensemble $\cS$, obtained by maximizing the measure-and-reconstruct fidelity function over orthonormal bases $\cB$. Then, solving for the optimal cloning fidelity reduces to the following simple form.
\begin{align}
& \cF_{\rm opt}(\cS) \label{eq:fopt_q} \\
& = \max_{p,q}\left[ 2pq +  (d-1)q^{2} + \left(\frac{A_{\rm opt}(\cS)}{M} \right)(p^{2}-2pq) \right] \nonumber \\
&= \max_{p,q}\left[\left(1 - \frac{A_{\rm opt}(\cS)}{M} \right) 2pq +  (d-1)q^{2} + \frac{A_{\rm opt}(\cS)}{M} (p^{2}) \right] \nonumber \\
&= \max_{q}\bigg[ 2q\sqrt{1-2(d-1)q^{2}} + (d-1)q^{2} \nonumber \\
& + \left(\frac{A_{\rm opt}(\cS)}{M} \right)(1 -2(d-1)q^{2}-2q\sqrt{1-2(d-1)q^{2}} )\bigg] , \nonumber
\end{align}
where we have used the relation $p^2+2(d-1)q^2=1$.

Since $p\geq 0$, we have, $0\leq q^2\leq \frac{1}{2(d-1)}$. Setting $\sqrt{2(d-1)}q=\sin \theta$, the optimization problem reduces to,

\begin{equation}
\cF_{\rm opt}(\cS) =  \max_{\theta} f(\theta),
\end{equation}
where the objective function $f(\theta)$ is defined as,
\begin{eqnarray}
f(\theta) &=& 
\bigg[ \frac{\sin^{2}\theta}{2} + \frac{\sin2\theta}{\sqrt{2(d-1)}}  \nonumber \\
&+& \frac{A_{\rm opt}(\cS)}{M}\left(\cos^{2}\theta - \frac{\sin2\theta}{\sqrt{2(d-1)}} \right) \bigg] . \label{eq:fopt_theta}
\end{eqnarray}

Therefore,
\begin{align}
\frac{d f}{d \theta} &=& \sin2\theta\left(\frac{1}{2} - \frac{A_{\rm opt}}{M}\right) + \frac{2\cos2\theta}{\sqrt{2(d-1)}}\left(1 - \frac{A_{\rm opt}}{M}\right) .\nonumber
\end{align}

Setting $\frac{d f}{d \theta}=0$, we see that the extremal values of $\theta$ occur as solutions to the equation,


\begin{equation}
\tan 2\theta =  \frac{4(A_{\rm opt} - M)}{(M-2A_{\rm opt})\sqrt{2(d-1)}}. \label{eq:theta_opt}
\end{equation}
The right hand side is a function of the number of states $M$ and the dimension $d$. Defining the function $\cG(M,d)$ as,
\begin{equation}
\cG(M,d) = \frac{4(A_{\rm opt} - M)}{(M-2A_{\rm opt})\sqrt{2(d-1)}}, \label{eq:G}
\end{equation}
the extremal values of the parameter $q$ are obtained by solving,
\begin{equation}
 \frac{2\sqrt{2(d-1)}q\sqrt{1-2(d-1)q^2}}{1-4(d-1)q^2} = \cG(M,d) . \label{eq:qopt1}
 \end{equation}
The optimal values of $q$ for which the cloning fidelity attains its extremal values are thus,
\begin{equation}
q_{\rm opt}=\pm\frac{1}{2\sqrt{d-1}}\sqrt{1 \pm \frac{1}{\sqrt{1+ (\cG(M,d))^2}}} . \label{eq:qopt2}
\end{equation}

Since the expression for average cloning fidelity in Eq.~\eqref{eq:fopt_q} is a sum of positive terms ($A_{\rm opt} \leq M$) with a linear term in $p,q$, we see that the maximum value of cloning fidelity is attained for $p\geq0$, $q\geq 0$. Thus the optimal values of $q$ for which the cloning fidelity is maximised are,
\begin{equation}
q_{\rm opt} = \frac{1}{2\sqrt{d-1}}\sqrt{1 \pm \frac{1}{\sqrt{1+\cG(N,d)^2}}}.
\end{equation}
By working out the second derivative of the function in Eq.~\eqref{eq:fopt_theta},
\begin{eqnarray}
\frac{d^{2} f}{d\theta^2} &=& 2\cos2\theta \bigg[\left(\frac{1}{2} - \frac{A_{\rm opt}}{M}\right) \nonumber \\
&-& \frac{2\tan2\theta}{\sqrt{2(d-1)}}\left(1 - \frac{A_{\rm opt}}{M}\right)\bigg], \label{eq:second_der}
\end{eqnarray}
we see that,
\begin{equation}
q_{\rm opt} = \frac{1}{2\sqrt{d-1}}\sqrt{1 - \frac{{\rm sgn} \left(\frac{A_{\rm opt}}{M} -\frac{1}{2}\right)}{\sqrt{1 + (\cG(M,d))^{2}}}}, \label{eq:qopt3}
\end{equation}
where ${\rm sgn}(x)$ is the signum function defined as 
\[{\rm sgn}(x) = \bigg\lbrace \begin{array}{cc} 
-1 & {\rm if} \; x < 0, \\
0 & {\rm if} \; x = 0, \\
+1 & {\rm if} \; x > 0. 
\end{array} \]

Our solution for the optimal cloning fidelity for any given ensemble of quantum states, leads us to observe some interesting properties of optimal QCMs. The first observation we may make is that, when the ensemble $\cS$ is such that $\frac{A_{\rm opt}}{M} = \frac{1}{2}$, the corresponding optimal QCM has parameters,
\begin{equation}
q_{\rm opt} = \frac{1}{2\sqrt{d-1}}, \; p_{\rm opt} = \frac{1}{\sqrt{2}}. \label{eq:Aopt_half}
\end{equation}

\subsection{Optimal Cloning Fidelity for $N$ MUBs in $d$ dimensions: Proof of Lemma~\ref{lem:Qc_MUB}}\label{sec:Qc_mub}

We now evaluate the optimal cloning fidelity for an ensemble $\cS_{\rm MUB}$ of states which constitute a set of MUBs. This in turn, provides an upper bound on the cloning-based incompatibility measure $\cQ_{c}$ as shown in Sec.~\ref{sec:cloning_measure}.  

Let $\cX_{\rm MUB} \equiv \{X^{1}, X^{2}, \ldots, X^{N}\}$ denote a set of $N$ MUBs in $d$-dimensions. The corresponding ensemble of eigenstates is given by $\cS_{\rm MUB} \equiv \{\frac{1}{Nd}, |\psi^{l}_{m}\rangle\langle \psi^{l}_{m}|\}$, with $1\leq l\leq N$ and $1\leq m \leq d$, where $|\psi^{l}_{m}\rangle$ is the $m^{\rm th}$ basis vector of the $l^{\rm th}$ basis $X^{l}$. 

As before, the optimal cloning fidelity for the ensemble $\cS_{\rm MUB}$ is calculated as,
\[ \cF_{\rm opt}(\cS_{\rm MUB})  = \max_{p,q: p^{2} > 2pq}\max_{\cB} \cF_{\rm avg}(\cS_{\rm MUB},\cB, p, q). \]
From the form of the average fidelity written down in Eq.~\eqref{eq:avg_fid2}, we see that,
\begin{eqnarray}
\cF_{\rm avg}(\cS_{\rm MUB}, \cB, p, q) &=& 2pq +  (d-1)q^{2} \nonumber \\
&+& \left(\frac{A(\cS_{\rm MUB}, \cB)}{Nd} \right)(p^{2}-2pq). \nonumber
\end{eqnarray}
The average cloning fidelity $\cF_{\rm avg}(\cS_{\rm MUB},\cB,p, q)$ depends on cloning basis $\cB$ only via the term $A(\cS, \cB)/Nd$. Furthermore, from Prop.~\ref{prop:zero_q} it follows that,
\[ \max_{\cB} \frac{A(\cS_{\rm MUB}, \cB)}{Nd} \equiv \max_{\cB}\mathbb{F}_{\rm avg}(\cS_{\rm MUB}, \cM_{\cB}, \cA_{\cB}),\]
where the RHS is the maximum average fidelity attained via a measure-and-reconstruct strategy, as defined in Sec.~\ref{sec:Q_avgFid}. The problem of finding the optimal cloning basis $\cB$ is thus reduced to the problem of finding the optimal projective measurement $\cM_{\cB}$. For an ensemble of MUBs, Lemma~\ref{lem:proj_opt} shows that the optimal basis is simply one of the bases in the set $\cX$.  Without loss of generality, we may therefore fix the cloning basis $\cB$ to be one of the eigenbases of one of the observables in the set, say $X^1$. 

Expanding the state $\ket{\psi^{l}_{m}}$ in the basis $\cB = \{\ket{\psi^{1}_{i}}, \, 1\leq i\leq d\}$, we have,
\[ \ket{\psi^{l}_{m}} = \sum^d_{i} (\alpha^{l}_{m})_i\ket{\psi_{i}^{1}} . \]

Since the observables in the set are mutually unbiased, the coefficients are such that $\mid  (\alpha ^{l}_{m})_{i} \mid=  \frac{1}{\sqrt{d}}$ for $l\neq 1$, and, $\mid (\alpha^{l}_{m})_{i} \mid = \delta_{im}$ for $l = 1$. Recalling the definitions of $A^{l}_{m}$ and $B^{l}_{m}$ from Eq.~\eqref{eq:A_B} in Sec.~\ref{sec:opt_cloning}, we have, 
\begin{eqnarray} 
A^{1}_{m} &=& \sum_i\mid  (\alpha^{1}_{m})_i\mid^4 =\sum_i\delta_{mi}=1 \nonumber \\
B^{1}_{m}&=& \sum_{i\neq j}\mid (\alpha^{1}_{m})_i \mid^2 \mid (\alpha^{1}_{m})_j \mid^2 = \sum_{i\neq j}\delta_{mi}^2\delta_{mj}^2 .\nonumber
\end{eqnarray}
For $l\neq 1$, 
\begin{eqnarray}
A^{l}_{m} &=& \sum_i\mid \alpha^{lm}_i\mid^4 
=\frac{1}{d} \nonumber \\
B^{l}_{m}&=& \sum^d_{{i,j=1},{i\neq j}} \mid (\alpha^{l}_{m})_i\mid^2 \mid (\alpha^{l}_{m})_j\mid^2 = \frac{(d-1)}{d} . \nonumber
\end{eqnarray}
Thus we have,
\[
 A_{\rm opt}(\cS_{\rm MUB}) = \sum_{l,m}A^{l}_{m} = N-1+d .\] 

Using this value of $A_{\rm opt}$, we see that for $q=0$,
\begin{equation}
\cF_{\rm opt}(\cS_{\rm MUB}) = \frac{A_{\rm opt}}{Nd} = \frac{(N-1+d)}{Nd},  \label{eq:Fproj_mub}
\end{equation}
which matches with the optimal fidelity obtained for an ensemble of MUBs using a measure-and-reconstruct strategy, obtained in~\cite{incompatibility_BM}.

Following the steps in Sec.~\ref{sec:Qc_gen}, we see that the cloning fidelity for the ensemble $\cS_{\rm MUB}$ is attained for 
\[
q_{\rm opt} = \frac{1}{2\sqrt{d-1}}\sqrt{1- \frac{{\rm sgn}\left(\frac{A_{\rm opt}}{Nd} - \frac{1}{2}\right)}{\sqrt{1+(\cG(N,d))^2}} } . 
\]
Hence,
\begin{equation}
\cF_{\rm opt}(\cS_{\rm MUB}) =\cF_{\rm avg}\left(\cS_{\rm MUB}, q_{\rm opt}, X^{1} \right)  
\end{equation}
is the optimal cloning fidelity that can be achieved using a  symmetric QCM, for a set of $N$ MUBs in $d$ dimensions.

\subsection{Optimal cloning fidelity for $N>d+1$ observables}\label{sec:N_d+1}

Since $\frac{A_{\rm opt}}{Nd} \geq \frac{2}{d+1}$, we have,
\begin{eqnarray}
&& \cF_{\rm avg}(\cS,\cB_{\rm opt}, p,q) \geq 2pq + (d-1)q^{2} + \frac{2}{d+1}(p^{2} - 2pq) \nonumber  \\
&=& \frac{1}{d+1}\left[ (d-1)2pq + 2p^{2} + (d+1)(d-1)q^{2} \right] \nonumber \\
&=& \frac{1}{d+1}\left[ 2 + (d-1)(d-3)q^{2} + 2(d-1)q\sqrt{1-2(d-1)q^{2}} \right] \nonumber \\
&\equiv& \left(\frac{1}{d+1}\right) h(q),  \nonumber 
\end{eqnarray}
where we have used $h(q)$ to denote the function to be maximised. Setting $\sqrt{2(d-1)}q = \sin\theta$, the optimization problem becomes,
\begin{eqnarray}
&& \max_{q}h(q)  \nonumber \\
&=& \max_{\theta}\left[ 2 + \left(\frac{d-3}{2}\right)\sin^{2}\theta + \sqrt{2(d-1)}\sin\theta\cos\theta \right] \nonumber \\
&=& \max_{\theta} h(\theta). \label{eq:hq}
\end{eqnarray}
Setting $\frac{dh}{d\theta} = 0$, we have,
\[
\left(\frac{d-3}{2}\right)\sin2\theta + \sqrt{2(d-1)}\cos2\theta = 0. \]
Thus the optimal value of $\theta$, denoted as $\theta_{\rm opt}$, satisfies,
\begin{equation}
\tan2\theta_{\rm opt} = \frac{-2\sqrt{2(d-1)}}{d-3} .\label{eq:theta_opt_d+1}
\end{equation}
To check that $h(\theta)$ attains a minimum at $\theta = \theta_{\rm opt}$, we note that the second derivative is,
\[h'' (\theta) = (d-3)\cos2\theta\left[ 1 - \left(\frac{2\sqrt{2(d-1)}}{d-3}\right) \right]\tan2\theta .\]
Setting $\theta=\theta_{\rm opt}$, we see that $h''(\theta_{\rm opt}) < 0$, for all $d>1$.
Simplifying Eq.~\eqref{eq:theta_opt_d+1}, and replacing $\theta$ with $q$, we have,
\begin{eqnarray}
\frac{\sin^{2}2\theta_{\rm opt}}{\cos^{2}2\theta_{\rm opt}} &=& \frac{8(d-1)}{(d-3)^{2}} \nonumber \\
\Rightarrow \cos2\theta_{\rm opt} &=& \pm \frac{d-3}{d+1} . \nonumber
\end{eqnarray}
We see that $h''(\theta_{\rm opt}) < 0$ for for all $\cos2\theta_{\rm opt}=-\frac{d-3}{d+1}$ and $h''(\theta_{\rm opt}) > 0$ otherwise for $d\neq3$. 

Substituting back for $q$, we have,
\begin{eqnarray}
1 - 4(d-1)q^{2}_{\rm opt} &=& -\frac{d-3}{d+1} \nonumber \\
\Rightarrow q_{\rm opt} &=& \sqrt{\frac{1}{2(d+1)}} . \label{eq:qopt_d+1}
\end{eqnarray}
Substituting back for $q_{\rm opt}$ in $h(q)$ (see Eq.~\eqref{eq:hq} above), we get,
\begin{eqnarray}
\cF_{\rm opt}(\cS) &\geq& \frac{1}{d+1}\left[ 2 + \frac{(d-3)(d-1)}{2(d+1)}  + \frac{2(d-1)}{d+1} \right] \nonumber \\
&=& \frac{d+3}{2(d+1)}  \equiv \tilde{\cF}_{\rm opt}(d),
\end{eqnarray}
as desired. For $d=3$, $q_{\rm opt} = \frac{1}{2\sqrt{2}}$ consistent with Eq.~\eqref{eq:qopt_d+1}. 

\section{Average cloning fidelity for a pair of qubit observables}\label{sec:qubit_app}

Recall from Lemma~\ref{lem:Qc_gen} that the incompatibility of a set of observables $\chi$, is obtained from the average cloning fidelity of the corresponding ensemble of eigenstates $\cS$ as, 
\[\cQ_{c}(\cX) = 1 - \cF_{\rm avg} (\cS,\cB_{\rm opt}, p_{\rm opt}, q_{\rm opt} ).\]
For the case of the qubit observables defined in Sec.~\ref{sec:qubit}, the average fidelity is evaluated from Eqn.~\ref{eq:avg_fid2} as,
\begin{eqnarray}
&& \cF_{\rm avg}(\cS) \nonumber \\
&=& \frac{\cA_{\rm opt}}{Nd}(p_{\rm opt}^2-2p_{\rm opt}q_{\rm opt})+2p_{\rm opt}q_{\rm opt}+(d-1)q_{\rm opt}^2 . \nonumber
\end{eqnarray}
For the case of the qubit observables, with the eigenstate ensemble $\cS_{2}$ given in Eq.~\eqref{eq:S_qubit}, we know from Lemma~\ref{lem:qubit} that, 
\begin{eqnarray}
q_{\rm opt}&=&\frac{1}{2}\sqrt{1-\frac{1}{1+\cG(\vec{a}.\vec{b})^2}}, \nonumber \\
\cG(\vec{a}.\vec{b})&=& \frac{\sqrt{2}(1-\mid\vec{a}.\vec{b}\mid)}{1+\mid\vec{a}.\vec{b}\mid)}, \nonumber
\end{eqnarray}
so that,
\begin{equation}
q_{\rm opt}=\frac{1}{2}\sqrt{1-\frac{1+\vec{a}.\vec{b}}{\sqrt{(1+\vec{a}.\vec{b})^2+2(1-\vec{a}.\vec{b})^2}}} .
\end{equation}
The normalisation condition $ p_{\rm opt}^2+2(d-1)q_{\rm opt}^2=1$ gives,

\begin{align}
p_{\rm opt} =\sqrt{\frac{1}{2}\left(1+\frac{1+\vec{a}.\vec{b}}{\sqrt{(1+\vec{a}.\vec{b})^2+2(1-\vec{a}.\vec{b})^2}}\right)} .
\end{align}


Finally, as shown in Eq.~\eqref{eq:A_qubit} (Sec.~\ref{sec:qubit}), we have,
\begin{align}
\cA_{\rm opt}(\cS_{2}) =3+\mid\vec{a}.\vec{b}\mid
\end{align}
Thus, the average cloning fidelity for a pair of qubit observables ($N=2$ and $d=2$), can be evaluated as,
\begin{eqnarray}
&& \cF_{\rm avg}(\cS_{2}) \nonumber \\
&=& \frac{5+\mid\vec{a}.\vec{b}\mid}{8}+\frac{1}{8}\sqrt{(1+\mid\vec{a}.\vec{b}\mid)^2+2(1-\mid\vec{a}.\vec{b}\mid)^2}.
\end{eqnarray}
%

If we take $x=\mid\vec{a}.\vec{b}\mid,(0 \leq x \leq 1)$, the graph is monotonically increasing with $x$.

\end{document}